\documentclass[12pt]{article}%
\usepackage{amsmath}
\usepackage{amsfonts}
\usepackage{amssymb}
\usepackage{graphicx}%
\providecommand{\U}[1]{\protect\rule{.1in}{.1in}}
\newtheorem{theorem}{Theorem}

\newtheorem{corollary}[theorem]{Corollary}

\newtheorem{definition}[theorem]{Definition}

\newtheorem{proposition}[theorem]{Proposition}
\newtheorem{remark}[theorem]{Remark}

\newenvironment{proof}[1][Proof]{\noindent\textbf{#1.} }{\ \rule{0.5em}{0.5em}}
\def \U {\overrightarrow U_{\!\!\lambda} }

\begin{document}

\title{A right inverse operator for $\operatorname{curl}+\lambda$ and applications}
\author{Briceyda B. Delgado, Vladislav V. Kravchenko$^{\text{*}}$\\{\small Regional mathematical center of Southern Federal University, }\\{\small Bolshaya Sadovaya, 105/42, Rostov-on-Don, 344006, Russia,}\\$^{\text{* on leave from Cinvestav, Mexico}}$\\{\small e-mail: briceydadelgado@gmail.com, vkravchenko@math.cinvestav.edu.mx
\thanks{Research was supported by CONACYT, Mexico via the project 284470.}}}
\maketitle

\begin{abstract}
A general solution of the equation $\operatorname{curl}\vec{w}+\lambda\vec
{w}=\overrightarrow{g},\,\lambda\in\mathbb{C},\,\lambda\neq0$ is obtained for
an arbitrary bounded domain $\Omega\subset\mathbb{R}^{3}$ with a Liapunov
boundary and $\overrightarrow{g}\in W^{p,\operatorname{div}}\left(
\Omega\right)  =\left\{  \overrightarrow{u}\in L^{p}\left(  \Omega\right)
:\,\operatorname{div}\overrightarrow{u}\in L^{p}\left(  \Omega\right)
,\,1<p<\infty\right\}  $. The result is based on the use of classical integral
operators of quaternionic analysis.

Applications of the main result are considered to a Neumann boundary value
problem for the equation $\operatorname{curl}\vec{w}+\lambda\vec
{w}=\overrightarrow{g}$ as well as to the nonhomogeneous time-harmonic Maxwell
system for achiral and chiral media.

\end{abstract}

\mbox{}\newline\noindent\textbf{Keywords:} Div-curl system, monogenic
functions, Helmholtz equation, metaharmonic\newline conjugate function,
Neumann boundary value problem, Maxwell's equations. \mbox{}\newline%
\noindent\textbf{Classification:} 30G20, 30G35, 35Q60.

\section{Introduction}

\label{sec:introduction} We study the nonhomogeneous equation
\begin{equation}
\operatorname{curl}\vec{w}+\lambda\vec{w}=\overrightarrow{g},\quad\lambda
\in\mathbb{C} \label{main eq Intro}%
\end{equation}
in a bounded domain $\Omega\subset\mathbb{R}^{3}$ and propose an integral
representation for its general solution or in other words we construct a
right-inverse operator for the operator $\operatorname{curl}+\lambda$ in an
appropriate functional space. In the situation when the boundary values of
$\vec{w}$ are given a representation of $\vec{w}$ is known (e.g., in the
context of quaternionic analysis it is obtained directly from the
Borel-Pompeiu formula \cite[pp. 59, 60]{Krav2003}). However, such
representation not always is convenient. It is clear, e.g., that the Newton
potential $L[u](\vec{x}):=-\frac{1}{4\pi}\int_{\Omega}\frac{{u(\vec{y})}%
}{\left\vert {\vec{x}-\vec{y}}\right\vert }{\,d\vec{y}}$ (the right inverse
operator for the Laplacian) is often used with no relation to concrete
boundary values of a solution to a Poisson equation.

The right-inverse operator for the operator $\operatorname{curl}+\lambda$ is
obtained by using a quaternionic approach. Equation (\ref{main eq Intro}) is
considered as a vector part of a quaternionic equation
\begin{equation}
(D+\lambda)\vec{w}=g \label{D+lambda Intro}%
\end{equation}
where $D$ is the Moisil-Teodorescu operator and $g_{0}=\operatorname*{Sc}g$ is
related with $\overrightarrow{g}$ by the equality $\operatorname{div}\vec
{g}+\lambda g_{0}=0$. A right inverse operator for $D+\lambda$ is well known
(see, e.g., \cite{GuSpr1990}, \cite{KravShap1996}, \cite{Krav2003}). Sometimes
it is called the Teodorescu transform and denoted by $T_{\lambda}$. However
$T_{\lambda}\left[  g\right]  $ is in general a complete quaternion (whose
scalar part is not necessarily zero), and, of course, simply resting a scalar
part from $T_{\lambda}\left[  g\right]  $ does not lead to a solution of
(\ref{D+lambda Intro}). Thus, the main problem for constructing a right
inverse for $\operatorname{curl}+\lambda$ reduces to finding a right inverse
for the operator $\left(  D+\lambda\right)  V$ where $V$ is a projection
operator $Vw=\operatorname*{Vec}\left(  w\right)  =\vec{w}$. This problem we
solve in three steps. First, we introduce certain component operators
conforming $T_{\lambda}$ and study their properties. Second, we give a
complete solution to the problem of constructing so-called metaharmonic
conjugate functions, considering the $(D+\lambda)w=0$ (for a full quaternion
$w$) construct $\overrightarrow{w}$ from a given $w_{0}=\operatorname*{Sc}w$
and vice versa, given $\overrightarrow{w}$ find $w_{0}$. Finally, with the aid
of these results we find out what term should be rested from $T_{\lambda
}\left[  g\right]  $ in order that the resulting function still be a solution
of the equation $(D+\lambda)w=g$ at the same time being purely vectorial.

The outline of this paper is as follows. In Section \ref{sec:Background} is
given the notation and some basic results on equation (\ref{D+lambda Intro}).
In Section \ref{subsec:Integral operators} we introduce a decomposition of the
$\lambda$-Teodorescu transform and some properties of the component operators.
In Section \ref{sec:Metaharmonic} the procedure for constructing metaharmonic
conjugate functions is presented. It is worth mentioning that in the case
$\lambda\neq0$ it resulted to be far more elementary and explicit than in the
case $\lambda=0$ (for which we refer to \cite[Prop.\ 2.3]{DelPor2017} and
\cite{DelPor2018}). In Section \ref{sec:Solution} the main result of this work
is presented which consists in a general solution of (\ref{main eq Intro}) and
an explicit expression for the right inverse operator for the operator
$\operatorname{curl}+\lambda$. In the rest of the paper we show some
applications of this result. In Section \ref{Sect Neumann bvp} a Neumann
problem for (\ref{main eq Intro}) is reduced to a boundary integral equation.
In Section \ref{sec:Maxwell's systems} a general weak solution of the
nonhomogeneous time-harmonic Maxwell system is obtained. In Section
\ref{Sect bvp Maxwell} it is applied to a standard boundary value problem for
the Maxwell system, well studied in the homogeneous case (e.g., in
\cite{Colton1983}) but not in the nonhomogeneous situation. Finally, in
Section \ref{sec:Maxwell-chiral}\ a general weak solution of the
nonhomogeneous Maxwell system for chiral media is presented.

\section{Background for the $D+\lambda$ system}

Together with the equation
\begin{equation}
\operatorname{curl}\vec{w}+\lambda\vec{w}=\overrightarrow{g} \label{main eq}%
\end{equation}
it is convenient to consider a quaternionic equation whose vectorial part
coincides with (\ref{main eq}). We begin by introducing the necessary notations.

\label{sec:Background} Let $\Omega\subset\mathbb{R}^{3}$ be a bounded domain.
We are interested in functions $w=w_{0}+\vec{w}\colon\Omega\rightarrow
\mathbb{B}$ with $w_{0}=\operatorname{Sc}w$, $\vec{w}=\operatorname{Vec}w$,
where $\mathbb{B}$ denotes the algebra of biquaternions.


From now on $\vec{x}\in\mathbb{R}^{3}$. The Moisil-Teodorescu differential
operator $D$ (also known as the generalized Cauchy-Riemann or occasionally the
Dirac operator but in fact was introduced by R. W. Hamilton) is defined by
\[
D=e_{1}\partial_{1}+e_{2}\partial_{2}+e_{3}\partial_{3},
\]
where $\partial_{i}=\partial/\partial x_{i}$, $i=1,2,3$ and $e_{i}$ stand for
basic quaternionic units. We remind that in terms of the classical
differential operators of vector calculus the action of $D$ can be written as
\[
Dw=-\operatorname{div}\vec{w}+\operatorname{grad}w_{0}+\operatorname{curl}%
\vec{w},
\]
meaning that $\operatorname{Sc}\left(  Dw\right)  =-\operatorname{div}\vec{w}$
and $\operatorname{Vec}\left(  Dw\right)  =\operatorname{grad}w_{0}%
+\operatorname{curl}\vec{w}$.


\begin{definition}
Let $w\in C^{1}(\Omega,\mathbb{B})$ and $\lambda\in\mathbb{C}$. We will say
that $w$ is $\lambda$-\textit{monogenic in }$\Omega$ if $w$ belongs to the
kernel of $D+\lambda$ \textit{in }$\Omega$.
\end{definition}

Or equivalently,%
\begin{equation}
(D+\lambda)w=0\iff\left\{
\begin{array}
[c]{rcl}%
\operatorname{div}\vec{w}\!\! & = & \!\!\lambda w_{0},\\
\operatorname{curl}\vec{w}+\lambda\vec{w}\!\! & = & \!\!-\operatorname{grad}%
w_{0}.
\end{array}
\right.  \label{system-alpha-monogenic}%
\end{equation}
When $\lambda=0$ the system (\ref{system-alpha-monogenic}) represents the
Moisil-Teodorescu system that defines the quaternionic monogenic functions
(see, e.g., \cite{GuSpr1990,GuHaSpr2008}).

If $w$ is $\lambda$-monogenic, it necessarily satisfies the \textit{Helmholtz
equation}
\begin{equation}
(\Delta+\lambda^{2})w=0. \label{eq:Helmholtz}%
\end{equation}

\begin{definition}
The purely vectorial $\lambda$-monogenic functions (when $\lambda\neq0$) are
called \textit{force-free fields} (or sometimes \textit{force-free magnetic
fields).}
\end{definition}

They satisfy the equation
\begin{equation}
\operatorname{curl}\vec{u}+\lambda\vec{u}=0 \label{eq:force-free-field}%
\end{equation}
which additionally implies that $\operatorname{div}\vec{u}=0$. Observe that
equation (\ref{eq:force-free-field}) implies that $\operatorname{curl}\vec
{u}\times\vec{u}=0$. Some references for the force-free fields are
\cite{Kress1977,Kress1978,Kress1981,Krav2005} and references therein.

We are especially interested in purely vectorial solutions of the
nonhomogeneous equation
\begin{equation}
(D+\lambda)\vec{w}=g \label{D+lambda}%
\end{equation}
where $g=g_{0}+\vec{g}\in L^{p}(\Omega,\mathbb{B})$, $1<p<\infty$. A purely
vectorial function $\vec{w}$ is a solution of (\ref{D+lambda}) iff it solves
the system of equations%
\begin{equation}
-\operatorname{div}\vec{w}=g_{0}, \label{system1}%
\end{equation}%
\begin{equation}
\operatorname{curl}\vec{w}+\lambda\vec{w}=\vec{g}. \label{system2}%
\end{equation}
The second equation of the system coincides with (\ref{main eq}) meanwhile the
first one is not independent. Indeed, application of $\operatorname{div}$ to
(\ref{system2}) leads to the compatibility condition $\operatorname{div}%
\vec{g}+\lambda g_{0}=0$. Defining the subspace of functions in $L^{p}%
(\Omega,\mathbb{B})$, $1<p<\infty$, where the system (\ref{system1}),
(\ref{system2}) is well-posed,
\[
\operatorname{Sol}_{\lambda}^{p}(\Omega):=\{g=g_{0}+\vec{g}\in L^{p}%
(\Omega,\mathbb{B})\colon\ \operatorname{div}\vec{g}+\lambda g_{0}=0\}
\]
we obtain that (\ref{main eq}) is equivalent to (\ref{D+lambda}) for
$g\in\operatorname{Sol}_{\lambda}^{p}(\Omega)$.

\begin{remark}
In the special case $\lambda=0$ the system (\ref{system1}), (\ref{system2})
reduces to the classical div-curl system. In \cite{DelPor2017} a general
solution for this first order partial differential system for star-shaped
domains was presented and in \cite{DelPor2018} for Lipschitz domains in
$\mathbb{R}^{3}$ with connected complement.

\end{remark}

Using the well-known \textit{generalized Green's formulas} \cite{Dautray1985}
a weak characterization of the solutions of (\ref{system1}), (\ref{system2})
is
\begin{equation}
\int_{\Omega}{(\operatorname{grad}v_{0})\vec{w}\,d\vec{y}}=\int_{\Omega
}{(\lambda\vec{w}-g)v_{0}\,d\vec{y}},\quad\forall v_{0}\in W_{0}^{1,q}%
(\Omega),\quad1/p+1/q=1.\label{weak-solution}%
\end{equation}
In particular, by (\ref{weak-solution}) the elements of $\operatorname{Sol}%
_{\lambda}^{p}(\Omega)$ satisfy the equality
\[
\int_{\Omega}{(\operatorname{grad}v_{0})\cdot\vec{g}\,d\vec{y}}={\lambda}%
\int_{\Omega}{g_{0}v_{0}\,d\vec{y}},\quad\forall v_{0}\in W_{0}^{1,q}(\Omega).
\]



\section{Some integral operators}

\label{subsec:Integral operators}

In this section we study the component operators of the $\lambda$-Teodorescu transform.

Using the fact that $\theta(\vec{x})=-e^{i\lambda|\vec{x}|}/(4\pi|\vec{x}|)$
is a fundamental solution of the Helmholtz operator $\Delta+\lambda^{2}$,
corresponding fundamental solutions of the operators $D\pm\lambda$ are given
by \cite[Th.\ 3.16]{KravShap1996}, \cite{Krav2003}
\[
E_{\pm\lambda}(\vec{x})=\pm\lambda\theta(\vec{x})-\operatorname{grad}%
\theta(\vec{x})=\theta(\vec{x})\left(  \pm\lambda+\frac{\vec{x}}{|\vec{x}%
|^{2}}-i\lambda\frac{\vec{x}}{|\vec{x}|}\right)  ,\quad\vec{x}\in
\mathbb{R}^{3}\setminus\{0\}.
\]

The $\lambda$-\textit{Teodorescu transform} 
is defined as follows (see, e.g., \cite{Krav2003})
\begin{align}
T_{\lambda}[w](\vec{x})  &  =\int_{\Omega}E_{\lambda}(\vec{x}-\vec{y}%
)w(\vec{y})\,d\vec{y},\quad\vec{x}\in\mathbb{R}^{3}%
.\label{eq:operador_Teodorescu}
\end{align}

Moreover, for every $w\in L^{p}(\partial\Omega,\mathbb{B})$, $1<p<\infty$,
$F_{\lambda}[w]$ is $\lambda$-monogenic and $T_{\lambda}$ is a right inverse
of $D+\lambda$.

\begin{proposition}
(\cite[Th.\ 4.14]{KravShap1996}) \label{prop:right-inverse} Let $\Omega$ be a
bounded domain with a Liapunov boundary, $w\in L^{p}(\Omega,\mathbb{B})$,
$p>1$. Then in the generalized sense
\[
(D+\lambda)T_{\lambda}[w](\vec{x})=w(\vec{x}),\quad\vec{x}\in\Omega.
\]

\end{proposition}

\begin{proposition}
\label{prop:alpha-harmonic} Let $\Omega$ be a bounded domain with a Liapunov
boundary, $w\in L^{p}(\Omega,\mathbb{B})$, $1<p<\infty$. Then

\begin{enumerate}
\item[(i)] $\operatorname{Sc}T_{\lambda}[w]$ is a solution of the Helmholtz
equation \eqref{eq:Helmholtz} if and only if $w\in\operatorname{Sol}_{\lambda
}^{p}(\Omega)$.

\item[(ii)] $\operatorname{Vec}T_{\lambda}[w]$ is a solution of the Helmholtz
equation (\ref{eq:Helmholtz}) if and only if $\operatorname{curl}\vec
{w}+\nabla w_{0}=\lambda\vec{w}$.
\end{enumerate}
\end{proposition}

\begin{proof}
Due to the factorization of the Helmholtz operator
\begin{equation}
\Delta+\lambda^{2}=-(D-\lambda)(D+\lambda) \label{fact}%
\end{equation}
and by Proposition \ref{prop:right-inverse}, we have that
\[
(\Delta+\lambda^{2})T_{\lambda}[w]=-(D-\lambda)[w]=\operatorname{div}\vec
{w}+\lambda w_{0}-\operatorname{curl}\vec{w}-\nabla w_{0}+\lambda\vec{w}%
\]
from where after separating the scalar and the vector parts the assertion follows.
\end{proof}

Following the decomposition of the Teodorescu transform for $\lambda=0$
\cite{DelPor2017} and using the relations $\operatorname{Sc}E_{\pm\lambda}%
=\pm\lambda\theta$ and $\operatorname{Vec}E_{\pm\lambda}=-\operatorname{grad}%
\theta$, let us define%
\begin{align*}
T_{0,\pm\lambda}[w](\vec{x})  &  =\int_{\Omega}\left[  \operatorname{Sc}%
E_{\pm\lambda}(\vec{x}-\vec{y})\,w_{0}(\vec{y})-\operatorname{Vec}%
E_{\pm\lambda}(\vec{x}-\vec{y})\cdot\vec{w}(\vec{y})\right]  \,d\vec{y}\\
&  =\int_{\Omega}\left[  \pm\lambda\theta(\vec{x}-\vec{y})w_{0}(\vec
{y})+\operatorname{grad}_{\vec{x}}\theta(\vec{x}-\vec{y})\cdot\vec{w}(\vec
{y})\right]  \,d\vec{y},
\end{align*}%
\begin{align*}
{\overrightarrow{T}}_{1,\pm\lambda}[w_{0}](\vec{x})  &  =\int_{\Omega
}\operatorname{Vec}E_{\pm\lambda}(\vec{x}-\vec{y})w_{0}(\vec{y})\,d\vec{y}\\
&  =-\int_{\Omega}\operatorname{grad}_{\vec{x}}\theta(\vec{x}-\vec{y}%
)w_{0}(\vec{y})\,d\vec{y}%
\end{align*}
and%
\begin{align*}
{\overrightarrow{T}}_{2,\pm\lambda}[\vec{w}](\vec{x})  &  =\int_{\Omega
}\left[  \operatorname{Sc}E_{\pm\lambda}(\vec{x}-\vec{y})\,\vec{w}(\vec
{y})+\operatorname{Vec}E_{\pm\lambda}(\vec{x}-\vec{y})\times\vec{w}(\vec
{y})\right]  \,d\vec{y}\\
&  =\int_{\Omega}\left[  \pm\lambda\theta(\vec{x}-\vec{y})\vec{w}(\vec
{y})-\operatorname{grad}_{\vec{x}}\theta(\vec{x}-\vec{y})\times\vec{w}(\vec
{y})\right]  \,d\vec{y}.
\end{align*}

Thus,
\begin{equation}
\label{decomposition}T_{\pm\lambda}[w]=T_{0,\pm\lambda}[w]+{\overrightarrow
{T}}_{\!\!1,\pm\lambda}[w_{0}]+\overrightarrow{T}_{\!\!2,\pm\lambda}[\vec{w}].
\end{equation}

Notice that ${\overrightarrow{T}}_{\!\!1,\lambda}={\overrightarrow{T}%
}_{\!\!1,-\lambda}$, meanwhile the difference for the other pair of operators
consists in a change of one sign. Under the hypothesis of Proposition
\ref{prop:alpha-harmonic}, we have that $T_{0,\lambda}[w]$ and
${\overrightarrow{T}}_{\!\!1,\lambda}[w_{0}]+\overrightarrow{T}_{\!\!2,\lambda
}[\vec{w}]$ are scalar and vector solutions of the Helmholtz equation
(\ref{eq:Helmholtz}), respectively.

Let us consider the Newton potential $L_{\lambda}\colon L^{p}(\Omega
)\rightarrow W^{2,p}(\Omega)$ defined by
\begin{equation}
L_{\lambda}[w](\vec{x})=\int_{\Omega}{\theta(\vec{x}-\vec{y})w(\vec{y}%
)\,d\vec{y}}\label{Newton-potential}%
\end{equation}
representing a right inverse for the Helmholtz operator $\Delta+\lambda^{2}$
(see, e.g., \cite[p. 155]{Babich et al 1964}). Using the fact that
$E_{\lambda}=-(D-\lambda)\theta$, we obtain the following relations.

\begin{proposition}
\label{prop:identities-Teodorescu} Let $\Omega$ be a bounded domain with a
Liapunov boundary and $w\in L^{p}(\Omega,\mathbb{B})$, $1<p<\infty$. Then%
\[
T_{0,\pm\lambda}[w](\vec{x})=\operatorname{div}L_{\lambda}[\vec{w}](\vec
{x})\pm\lambda L_{\lambda}[w_{0}](\vec{x}),
\]%
\[
{\overrightarrow{T}}_{1,\pm\lambda}[w_{0}](\vec{x})=-\operatorname{grad}%
L_{\lambda}[w_{0}](\vec{x})
\]
and%
\[
{\overrightarrow{T}}_{2,\pm\lambda}[\vec{w}](\vec{x})=-\operatorname{curl}%
L_{\lambda}[\vec{w}](\vec{x})\pm\lambda L_{\lambda}[\vec{w}](\vec{x}).
\]
Moreover, ${\overrightarrow{T}}_{1,\lambda}[w_{0}]$ is an irrotational vector
field and $-T_{0,\lambda}[\vec{w}]+\overrightarrow{T}_{\!\!2,\lambda}[\vec
{w}]\in\operatorname{Sol}_{\lambda}^{p}(\Omega)$.
\end{proposition}

The proof is straightforward.

The result of Proposition \ref{prop:identities-Teodorescu} can be summarized
in the following expression for the $\lambda$-Teodorescu operator $T_{\lambda
}[w]=-(D-\lambda)L_{\lambda}[w]$.

\section{Construction of metaharmonic conjugate functions}

\label{sec:Metaharmonic}

\begin{definition}
We will say that a scalar function $w_{0}$ and a vector function $\vec{w}$ are
metaharmonic conjugate functions to each other in $\Omega$ if the biquaternion
valued function $w:=w_{0}+\vec{w}$ is $\lambda$-monogenic in $\Omega$, that is
$(D+\lambda)(w_{0}+\vec{w})=0$.
\end{definition}

In this section we propose a procedure for constructing $\vec{w}$ when $w_{0}$
is known and vice versa.

\begin{remark}
The construction of the conjugate functions in the case $\lambda=0$ in general
is not an elementary operation. For star-shaped domains it can be performed
explicitly based on a radial integral operator \cite[Prop.\ 2.3]{DelPor2017}
(see also \cite{Sud1979} for star-shaped domains in $\mathbb{H}$), meanwhile
for a more general domain the procedure is far less explicit and can be
defined in terms of the Cauchy integral operator requiring the inversion of a
boundary integral operator (see \cite[Appendix]{DelPor2018} where this
procedure was introduced for bounded Lipschitz domains with a connected
complement). Both methods were fundamental in the solution of the div-curl
system provided in their respective references.

It is worth mentioning that the construction of the metaharmonic conjugate
functions (for $\lambda\neq0$) is far more simple and can be performed
explicitly in general.
\end{remark}

\begin{theorem}
\label{th:metaharmonic} (i) Let $w_{0}$ be a solution of $(\Delta+\lambda
^{2})w_{0}=0$ in a domain $\Omega$. Then the function $\vec{w}:=-\frac
{1}{\lambda}\nabla w_{0}$ is its \textit{metaharmonic conjugate}. It is
defined up to an arbitrary purely vectorial solution of $(D+\lambda
)\overrightarrow{v}=0$.

(ii) Let
\begin{equation}
\vec{w}\in\operatorname{Ker}(\Delta+\lambda^{2})\left(  \Omega\right)
.\label{ker}%
\end{equation}
Then a necessary and sufficient condition for the existence of a $\lambda
$-monogenic function $w$ in $\Omega$ such that $\operatorname{Vec}w=\vec{w}$
is the equality
\begin{equation}
\operatorname{curl}(\operatorname{curl}+\lambda)\vec{w}%
=0,\label{curlcurl cond}%
\end{equation}
and if it is satisfied then the unique \textit{metaharmonic conjugate to
}$\vec{w}$ has the form $w_{0}=\frac{1}{\lambda}\operatorname{div}%
\overrightarrow{w}$.
\end{theorem}

\begin{proof}
As a corollary of (\ref{fact}) we have that if $(\Delta+\lambda^{2})w_{0}=0$,
the function $w:=-\frac{1}{\lambda}\left(  D-\lambda\right)  w_{0}$ is a
solution of (\ref{system-alpha-monogenic}) and additionally
$\operatorname*{Sc}w=w_{0}$. Thus, (i) is proved.

Let us prove (ii). The necessity of (\ref{curlcurl cond}) follows directly
from (\ref{system-alpha-monogenic}). For the sufficiency, consider
$w_{0}=\frac{1}{\lambda}\operatorname{div}\overrightarrow{w}$. Hence the
scalar equation in (\ref{system-alpha-monogenic}) is satisfied. Application of
the gradient leads then to the equalities
\[
\operatorname{grad}w_{0}=\frac{1}{\lambda}\operatorname{grad}%
\operatorname{div}\overrightarrow{w}=\frac{1}{\lambda}\left(  \Delta
\overrightarrow{w}+\operatorname{curl}\operatorname{curl}\overrightarrow
{w}\right)  .
\]
Now, using (\ref{ker}) and (\ref{curlcurl cond}) we obtain
\[
\operatorname{grad}w_{0}=-\left(  \lambda\overrightarrow{w}%
+\operatorname{curl}\overrightarrow{w}\right)
\]
which is the vector equation in (\ref{system-alpha-monogenic}).
\end{proof}

\section{Solution of system (\ref{main eq Intro})}

\label{sec:Solution}

\begin{theorem}
\label{th:solution-div-curl} Let $\Omega$ be a bounded domain with a Liapunov
boundary, $\lambda\in\mathbb{C}$, $\lambda\neq0$ and $\overrightarrow{g}\in
W^{p,\operatorname{div}}\left(  \Omega\right)  $, $1<p<\infty$. Then a weak
general solution of the system (\ref{main eq}) is given by%
\begin{equation}
\vec{w}=\frac{1}{\lambda}\left(  \vec{g}-\operatorname{curl}\overrightarrow
{T}_{2,\lambda}[\vec{g}]\right)  +\vec{u}, \label{solution w}%
\end{equation}
where $\vec{u}$ is an arbitrary solution of (\ref{eq:force-free-field}).
\end{theorem}

\begin{proof}
Let $g_{0}:=-\frac{1}{\lambda}\operatorname{div}\vec{g}$. Then $g:=g_{0}%
+\vec{g}\in\operatorname{Sol}_{\lambda}^{p}(\Omega)$. Using Proposition
\ref{prop:alpha-harmonic} (i), we have that $T_{0,\lambda}[g]$ is a solution
of the Helmholtz equation \eqref{eq:Helmholtz}. Thus, by Theorem
\ref{th:metaharmonic} (i),
\[
(D+\lambda)\left(  T_{0,\lambda}[g]-\frac{\operatorname{grad}T_{0,\lambda}%
[g]}{\lambda}\right)  =0.
\]
Consider
\begin{equation}
\vec{w}:=T_{\lambda}[g]-T_{0,\lambda}[g]+\frac{\operatorname{grad}%
T_{0,\lambda}[g]}{\lambda}+\vec{u}. \label{vect w}%
\end{equation}
Since $T_{\lambda}$ is a right inverse of $D+\lambda$ (Proposition
\ref{prop:right-inverse}), we obtain
\[
(D+\lambda)\vec{w}=(D+\lambda)T_{\lambda}[g]-(D+\lambda)\left(  T_{0,\lambda
}[g]-\frac{\operatorname{grad}T_{0,\lambda}[g]}{\lambda}\right)  =g.
\]
Due to the decomposition \eqref{decomposition} of the Teodorescu transform,
the vector part of the last equality can be written as follows
\[
\operatorname{curl}\left({\overrightarrow{T}}_{\!\!1,\lambda}[g_{0}%
]+\overrightarrow{T}_{\!\!2,\lambda}[\vec{g}]\right)+\nabla T_{0,\lambda}[g]+\lambda \left({\overrightarrow{T}}_{\!\!1,\lambda}[g_{0}]+\overrightarrow{T}_{\!\!2,\lambda}[\vec{g}]\right)=\vec{g}.
\]
Thus,
\[
{\overrightarrow{T}}_{\!\!1,\lambda}[g_{0}]+\overrightarrow{T}_{\!\!2,\lambda
}[\vec{g}]+\frac{\operatorname{grad}T_{0,\lambda}[g]}{\lambda}=\frac
{1}{\lambda}\left(  \vec{g}-\operatorname{curl}\left({\overrightarrow{T}%
}_{\!\!1,\lambda}[g_{0}]+\overrightarrow{T}_{\!\!2,\lambda}[\vec{g}]\right)\right)
\]
or, equivalently,%
\[
T_{\lambda}[g]-T_{0,\lambda}[g]+\frac{\operatorname{grad}T_{0,\lambda}%
[g]}{\lambda}=\frac{1}{\lambda}\left(  \vec{g}-\operatorname{curl}%
\left({\overrightarrow{T}}_{\!\!1,\lambda}[g_{0}]+\overrightarrow{T}_{\!\!2,\lambda}[\vec{g}]\right)\right)  .
\]
Hence (\ref{vect w}) takes the form
\[
\vec{w}(\vec{x})=\frac{1}{\lambda}\left(\vec{g}-\operatorname{curl}%
\left({\overrightarrow{T}}_{\!\!1,\lambda}[g_{0}]+\overrightarrow{T}_{\!\!2,\lambda}[\vec{g}]\right)\right)  +\vec{u},
\]
By Proposition \ref{prop:identities-Teodorescu}, $\operatorname{curl}%
{\overrightarrow{T}}_{\!\!1,\lambda}[g_{0}]=0$, thus (\ref{solution w}) is obtained.
\end{proof}

\begin{corollary}
Under the hypothesis of Theorem \ref{th:solution-div-curl} the operator
\[
R_{\lambda}:=\frac{1}{\lambda}(I-\operatorname{curl}\overrightarrow
{T}_{\!\!2,\lambda}),
\]
is a right inverse of the operator $\operatorname{curl}+\lambda I$ on
$W^{p,\operatorname{div}}\left(  \Omega\right)  $, $1<p<\infty$.
\end{corollary}

\begin{remark}
Theorem \ref{th:solution-div-curl} allows one to solve also slightly more
general systems of the form
\begin{equation}
\operatorname{curl}\overrightarrow{v}+\lambda\overrightarrow{v}+\nabla
\varphi\times\overrightarrow{v}=\vec{h} \label{sl more gen}%
\end{equation}
where $\varphi$ is an arbitrary continuously differentiable scalar function.
Indeed, $\overrightarrow{v}$ is a solution of (\ref{sl more gen}) iff
$\overrightarrow{w}:=e^{\varphi}\overrightarrow{v}$ satisfies (\ref{main eq})
with $\overrightarrow{g}=e^{\varphi}\vec{h}$.
\end{remark}

\section{A Neumann boundary value problem\label{Sect Neumann bvp}}

\label{sec:Neumann problem} Following \cite{Kress1981} we consider the Neumann
problem for the equation $\operatorname{curl}\vec{w}+\lambda\vec{w}=\vec{g}$
in a domain with a connected boundary belonging to the class $C^{2}$. The
problem consists in finding $\vec{w}\in C^{1}(\Omega)\cap C(\overline{\Omega
})$ such that%
\begin{equation}
\operatorname{curl}\vec{w}+\lambda\vec{w}=\vec{g},\quad\text{ in }\Omega,
\label{Neumann 1}%
\end{equation}%
\begin{equation}
\vec{w}|_{\partial\Omega}\cdot\overrightarrow{n}=\varphi_{0},\quad\text{ on
}\partial\Omega, \label{Neumann 2}%
\end{equation}
where $\vec{g}\in C^{1,\gamma}(\Omega)$, $\varphi_{0}\in C^{0,\gamma}%
(\partial\Omega)$ and $0<\gamma<1$. A necessary condition for the existence of
a solution is
\begin{equation}
\int_{\partial\Omega}[\vec{g}|_{\partial\Omega}\cdot\overrightarrow{n}%
-\lambda\varphi_{0}]\,ds_{\vec{y}}=0. \label{eq:condition}%
\end{equation}

The next definition was introduced in \cite{Kress1981} in order to study the
Neumann problem (\ref{Neumann 1}), (\ref{Neumann 2}) using the relation with
the Helmholtz equation.

\begin{definition}
We will say that $\lambda$ is regular with respect to the Neumann problem
(\ref{Neumann 1}), (\ref{Neumann 2}) if for all solutions $\vec{w}\in
C^{1}(\Omega)\cap C(\overline{\Omega})$, $a\in C^{2}(\Omega)\cap
C(\overline{\Omega})$ of the system of differential equations
\begin{align*}
\operatorname{curl}\vec{w}+\lambda\vec{w}  &  =\vec{g}+\operatorname{grad}a,\\
\Delta a+\lambda^{2}a  &  =0,\quad\text{ in }\Omega,
\end{align*}
satisfying the boundary conditions
\[
\vec{w}|_{\partial\Omega}\cdot\overrightarrow{n}=\varphi_{0},\quad
a|_{\partial\Omega}=0,\quad\text{ on }\partial\Omega,
\]
there follows $a\equiv0$ in $\Omega$.
\end{definition}

Using the operator $R_{\lambda}$ we transform the boundary value problem
(\ref{Neumann 1}), (\ref{Neumann 2}) into a Neumann problem for force-free fields.

\begin{theorem}
Let $\Omega\subset\mathbb{R}^{3}$ be a bounded domain with a $C^{2}$-boundary.
Let $\lambda$ be regular with respect to (\ref{Neumann 1}), (\ref{Neumann 2}).
Then a solution of the Neumann problem (\ref{Neumann 1}), (\ref{Neumann 2}) is
given by
\[
\vec{w}=R_{\lambda}[\vec{g}]+\vec{u},
\]
where
\begin{align}
\vec{u}(\vec{x}) &  =-\operatorname{grad}\int_{\partial\Omega}\theta(\vec
{x}-\vec{y})\psi_{0}(\vec{y})\,ds_{\vec{y}}+(\operatorname{curl}-\lambda
)\int_{\partial\Omega}\theta(\vec{x}-\vec{y})\vec{\psi}(\vec{y})\,ds_{\vec{y}%
},\quad\vec{x}\in\Omega,\nonumber\label{eq:solution-Neumann-reduced}\\
\psi_{0} &  =\varphi_{0}-R_{\lambda}[\vec{g}]\Bigg|_{\partial\Omega}%
\cdot\overrightarrow{n},\quad\text{ on }\partial\Omega,
\end{align}
and $\vec{\psi}=\vec{u}|_{\partial\Omega}\times\overrightarrow{n}$ is a
continuous solution of the boundary integral equation
\begin{align}
&  \frac{1}{2}\vec{\psi}(\vec{x})+\int_{\partial\Omega}\overrightarrow{n}%
(\vec{x})\times\left(  \lambda\theta(\vec{x}-\vec{y})\vec{\psi}(\vec
{y})-\operatorname{grad}_{\vec{x}}\theta(\vec{x}-\vec{y})\times\vec{\psi}%
(\vec{y})\right)  \,ds_{\vec{y}}\nonumber\label{eq:boundary-integral}\\
&  =\overrightarrow{n}(\vec{x})\times\int_{\partial\Omega}\operatorname{grad}%
_{\vec{x}}\theta(\vec{x}-\vec{y})\psi_{0}(\vec{y})\,ds_{\vec{y}},\quad\vec
{x}\in\partial\Omega.
\end{align}

\end{theorem}

\begin{proof}
By Theorem \ref{th:solution-div-curl} we have that $\vec{w}=R_{\lambda}%
[\vec{g}]+\vec{u}$ satisfies (\ref{Neumann 1}), where $\vec{u}$ is an
arbitrary force-free field. Therefore solution of the Neumann problem
(\ref{Neumann 1}), (\ref{Neumann 2}) reduces to solution of the Neumann
problem for a force-free field with a corresponding boundary condition%
\begin{equation}
\operatorname{curl}\vec{u}+\lambda\vec{u}=0, \label{Neumann fff 1}%
\end{equation}%
\begin{equation}
\vec{u}|_{\partial\Omega}\cdot\overrightarrow{n}=\psi_{0}=\varphi
_{0}-R_{\lambda}[\vec{g}]\Bigg|_{\partial\Omega}\cdot\overrightarrow{n}.
\label{Neumann fff 2}%
\end{equation}
Since (\ref{eq:condition}) is assumed to be fulfilled we have that $\psi_{0}$
also satisfies \eqref{eq:condition},
\[
\int_{\partial\Omega}\lambda\psi_{0}\,ds_{\vec{y}}=\int_{\partial\Omega
}\operatorname{curl}\overrightarrow{T}_{\!\!2,\lambda}[\vec{g}]|_{\partial
\Omega}\cdot\overrightarrow{n}\,ds_{\vec{y}}=0
\]
(due to the divergence theorem). The rest of the proof consists in applying
the solution given in \cite[Th.\ 3.3]{Kress1981} to the Neumann problem for
force-free fields (\ref{Neumann fff 1}), (\ref{Neumann fff 2}).
\end{proof}

\section{Time-harmonic Maxwell's equations}

\label{sec:Maxwell's systems} Let us consider the \textit{time-harmonic
Maxwell equations}%
\begin{equation}
\operatorname{curl}\vec{H}=-i\omega\epsilon\vec{E}+\vec{j},\quad
\operatorname{div}\vec{H}=0, \label{Maxw 1}%
\end{equation}%
\begin{equation}
\operatorname{curl}\vec{E}=i\omega\mu\vec{H},\quad\operatorname{div}\vec
{E}=\frac{\rho}{\epsilon}, \label{Maxw 2}%
\end{equation}
for a homogeneous isotropic medium. The quantities $\epsilon$ and $\mu$ are
complex numbers. The \textit{wave number} $\lambda=\omega\sqrt{\epsilon\mu}$
is chosen such that $\operatorname{Im}\lambda\geq0$. The charge density and
the current density are related by the equality $\rho=\frac{1}{i\omega
}\operatorname{div}\vec{j}$. Some references to the theory of time-harmonic
Maxwell's equations are \cite{Colton1983,Colton1992,Muller1969}.

Following \cite{Krav1992} (see also \cite{KravShap1996} and \cite{Krav2003})
the Maxwell system can be diagonalized with the aid of a pair of purely
vectorial biquaternion valued functions $\vec{\varphi}:=-i\omega\epsilon
\vec{E}+\lambda\vec{H}$ and $\vec{\psi}:=i\omega\epsilon\vec{E}+\lambda\vec
{H}$, obtaining
\begin{equation}
(D-\lambda)\vec{\varphi}=\operatorname{div}\vec{j}+\lambda\vec{j}%
,\quad(D+\lambda)\vec{\psi}=-\operatorname{div}\vec{j}+\lambda\vec{j}.
\label{eq:Maxwell-quaternionic}%
\end{equation}
Notice that the pair of equations \eqref{eq:Maxwell-quaternionic} is
equivalent to the system (\ref{Maxw 1}), (\ref{Maxw 2}).

\begin{theorem}
\label{theo:Maxwell} Let $\Omega$ be a bounded domain with a Liapunov boundary
and $\vec{j}\in W^{p,\operatorname{div}}(\Omega)$, $1<p<\infty$. Then a
general weak solution of the time-harmonic Maxwell system (\ref{Maxw 1}),
(\ref{Maxw 2}) is given by
\begin{align}
\vec{E}  &  =\frac{1}{2i\omega\epsilon}\left(  2\vec{j}-\operatorname{curl}%
(\overrightarrow{T}_{\!\!2,\lambda}+\overrightarrow{T}_{\!\!2,-\lambda}%
)[\vec{j}]\right)  +\frac{1}{2i\omega\epsilon}\left(  \vec{u}-\vec{v}\right)
,\nonumber\label{eq:solution-Maxwell}\\
\vec{H}  &  =-\frac{1}{2\lambda}\left(  \operatorname{curl}(\overrightarrow
{T}_{\!\!2,\lambda}-\overrightarrow{T}_{\!\!2,-\lambda})[\vec{j}]\right)
+\frac{1}{2\lambda}\left(  \vec{u}+\vec{v}\right)  ,
\end{align}
where $\vec{u}$ and $\vec{v}$ are arbitrary force-free fields associated to
the wave numbers $\lambda$ and $-\lambda$, respectively. That is,
$(D+\lambda)\vec{u}=0$ and $(D-\lambda)\vec{v}=0$.
\end{theorem}

\begin{proof}
Notice that if $\vec{j}\in W^{p,\operatorname{div}}(\Omega)$, then for the
right hand side of \eqref{eq:Maxwell-quaternionic} we have
$-\operatorname{div}\vec{j}+\lambda\vec{j}\in\operatorname{Sol}_{\lambda}%
^{p}(\Omega)$ and $\operatorname{div}\vec{j}+\lambda\vec{j}\in
\operatorname{Sol}_{-\lambda}^{p}(\Omega)$. Thus, by Theorem
\ref{th:solution-div-curl} we have that
\[
\vec{\varphi}=-\vec{j}+\operatorname{curl}\overrightarrow{T}_{\!\!2,-\lambda
}[\vec{j}]+\vec{v},\quad\vec{\psi}=\vec{j}-\operatorname{curl}\overrightarrow
{T}_{\!\!2,\lambda}[\vec{j}]+\vec{u},
\]
are weak solutions of \eqref{eq:Maxwell-quaternionic}, where $\vec{u}$ and
$\vec{v}$ are arbitrary force-free fields associated to the wave numbers
$\lambda$ and $-\lambda$, respectively.

Now \eqref{eq:solution-Maxwell} is obtained by noting that the pairs of vector
fields $(\vec{E},\vec{H})$ and $(\vec{\varphi},\vec{\psi})$\ are related by
the equations $2i\omega\epsilon\vec{E}=\vec{\psi}-\vec{\varphi}$ and
$2\lambda\vec{H}=\vec{\psi}+\vec{\varphi}$.
\end{proof}

\begin{remark}
Using Proposition \ref{prop:identities-Teodorescu} to compute $\overrightarrow
{T}_{\!\!2,\lambda}+\overrightarrow{T}_{\!\!2,-\lambda}$ and $\overrightarrow
{T}_{\!\!2,\lambda}-\overrightarrow{T}_{\!\!2,-\lambda}$ one can write
\eqref{eq:solution-Maxwell} in the form
\begin{align}
\vec{E}  &  =\frac{1}{i\omega\epsilon}\left(  \vec{j}+\operatorname{curl}%
\operatorname{curl}L_{\lambda}[\vec{j}]\right)  +\frac{1}{2i\omega\epsilon
}\left(  \vec{u}-\vec{v}\right)  ,\nonumber\label{eq:solution-Maxwell-3}\\
\vec{H}  &  =-\operatorname{curl}L_{\lambda}[\vec{j}]+\frac{1}{2\lambda
}\left(  \vec{u}+\vec{v}\right)  ,
\end{align}
where $L_{\lambda}$ is the right inverse of the operator $\Delta+\lambda^{2}$
defined in \eqref{Newton-potential}.

By (\ref{decomposition}), we have
\begin{align*}
(\overrightarrow{T}_{\!\!2,\lambda}+\overrightarrow{T}_{\!\!2,-\lambda}%
)[\vec{j}]  &  =-2\int_{\Omega}\operatorname{grad}_{\vec{x}}\theta(\vec
{x}-\vec{y})\times\vec{j}(\vec{y})\,d\vec{y},\\
(\overrightarrow{T}_{\!\!2,\lambda}-\overrightarrow{T}_{\!\!2,-\lambda}%
)[\vec{j}]  &  =2\lambda\int_{\Omega}\theta(\vec{x}-\vec{y})\vec{j}(\vec
{y})\,d\vec{y}.
\end{align*}
Therefore the weak solution \eqref{eq:solution-Maxwell} of the time-harmonic
Maxwell's system can be rewritten as follows
\begin{align}
\vec{E}  &  =\frac{1}{i\omega\epsilon}\left(  \vec{j}+\operatorname{curl}%
\int_{\Omega}\operatorname{grad}_{\vec{x}}\theta(\vec{x}-\vec{y})\times\vec
{j}(\vec{y})\,d\vec{y}\right)  +\frac{1}{2i\omega\epsilon}\left(  \vec{u}%
-\vec{v}\right)  ,\nonumber\label{eq:solution-Maxwell-2}\\
\vec{H}  &  =-\operatorname{curl}\left(  \int_{\Omega}\theta(\vec{x}-\vec
{y})\vec{j}(\vec{y})\,d\vec{y}\right)  +\frac{1}{2\lambda}\left(  \vec{u}%
+\vec{v}\right)  .
\end{align}
Compare with the solution given in \cite[p.\ 62]{Krav2003} where the boundary
values of $\vec{E}$ and $\vec{H}$ are assumed to be known.
\end{remark}

\section{Boundary value problems for the time-harmonic Maxwell
equations\label{Sect bvp Maxwell}}

\label{sec:Boundary-Maxwell} With the aid of Theorem \ref{theo:Maxwell} the
method of integral equations developed for boundary value problems for
homogeneous Maxwell equations (see, e.g., \cite[Ch.\ 4]{Colton1983}) can be
extended onto the nonhomogeneous equations. As an example, we study the
following boundary value problem. Find a solution of the Maxwell system
(\ref{Maxw 1}), (\ref{Maxw 2}) provided with the boundary condition
\begin{equation}
\vec{E}|_{\partial\Omega}\times\overrightarrow{n}=\vec{\varphi}.
\label{eq:boundary-Maxwell}%
\end{equation}
The system is considered in a bounded domain $\Omega$ with a $C^{2}$-boundary,
$\vec{j}\in C^{0,\gamma}(\overline{\Omega})$, $0<\gamma<1$ and $\vec{\varphi
}\in\mathcal{F}_{\text{Div}}(\partial\Omega):=\{\vec{\phi}\in C^{0,\gamma
}(\partial\Omega)\colon\operatorname{Div}\vec{\phi}\in C^{0,\gamma}%
(\partial\Omega)\}$. That is the surface divergence of $\vec{\varphi}$ (see
\cite[Def. 2.28]{Colton1983}) exists and belongs to the H\"{o}lder space
$C^{0,\gamma}(\partial\Omega)$.

Analogously to the procedure used in Section \ref{sec:Neumann problem} this
boundary value problem is transformed into a boundary value problem for a
homogeneous Maxwell system. Denote
\[
\vec{E}^{\ast}=i(\vec{u}-\vec{v}),\quad\vec{H}^{\ast}=\vec{u}+\vec{v},
\]
where $\vec{u}$ and $\vec{v}$ are arbitrary force-free fields from
\eqref{eq:solution-Maxwell}. Since $\operatorname{curl}(\vec{u}\pm\vec
{v})=-\lambda(\vec{u}\mp\vec{v})$, it is immediate that $(\vec{E}^{\ast}%
,\vec{H}^{\ast})$ satisfy the homogeneous time-harmonic Maxwell equations. By
Theorem \ref{theo:Maxwell} and \eqref{eq:solution-Maxwell-3}, the boundary
value problem (\ref{Maxw 1}), (\ref{Maxw 2}), \eqref{eq:boundary-Maxwell} is
equivalent to finding a pair of vector fields $(\vec{E}^{\ast},\vec{H}^{\ast
})\in C^{1}(\Omega)\cap C(\overline{\Omega})$ satisfying%
\begin{equation}
\operatorname{curl}\vec{E}^{\ast}+i\lambda\vec{H}^{\ast}=0,\quad
\operatorname{curl}\vec{H}^{\ast}-i\lambda\vec{E}^{\ast}=0, \label{Maxw hom}%
\end{equation}%
\begin{equation}
\vec{E}^{\ast}|_{\partial\Omega}\times\overrightarrow{n}=-2i\left(  \vec
{j}+\operatorname{curl}\operatorname{curl}L_{\lambda}[\vec{j}]\right)
\Bigg|_{\partial\Omega}\times\overrightarrow{n}-2\omega\epsilon\vec{\varphi}.
\label{eq:homo-Maxwell}%
\end{equation}

\begin{theorem}
Let $\Omega$ be a bounded domain with a $C^{2}$-boundary. Let $\vec{\varphi
}=\vec{E}|_{\partial\Omega}\times\overrightarrow{n}\in\mathcal{F}_{\text{Div}%
}(\partial\Omega)$ and $\vec{j}|_{\partial\Omega}\in\mathcal{F}_{\text{Div}%
}(\partial\Omega)$. If $\operatorname{Im}\lambda>0$, then there exists at most
one solution of the Maxwell boundary value problem (\ref{Maxw 1}),
(\ref{Maxw 2}), \eqref{eq:boundary-Maxwell}.
\end{theorem}

\begin{proof}
Due to the reduction of the problem (\ref{Maxw 1}), (\ref{Maxw 2}),
\eqref{eq:boundary-Maxwell} to the homogeneous problem (\ref{Maxw hom}),
(\ref{eq:homo-Maxwell}) we just need to verify that $\vec{E}^{\ast}%
|_{\partial\Omega}\times\overrightarrow{n}$ belongs to $\mathcal{F}%
_{\text{Div}}(\partial\Omega)$ and then application of \cite[Sec.
4.3]{Colton1983} gives us the result. Let $\{S_{n}\}$ be a sequence of
surfaces contained in $\partial\Omega$ with boundaries $\partial S_{n}$ of
class $C^{2}$ that converges to the point $\vec{x}\in\partial\Omega$ in the
sense of \cite[Def. 2.28]{Colton1983} and $\overrightarrow{n}_{n}$ the outward
unit normal vector to $\partial S_{n}$. Then
\begin{align*}
\operatorname{Div}(\operatorname{curl}\operatorname{curl}L_{\lambda}[\vec
{j}])(\vec{x})  &  =\lim_{S_{n}\rightarrow\vec{x}}\frac{1}{|S_{n}|}%
\int_{\partial S_{n}}\overrightarrow{n}_{n}\cdot\operatorname{curl}%
\operatorname{curl}L_{\lambda}[\vec{j}]\,ds_{\vec{y}}\\
&  =\lim_{S_{n}\rightarrow\vec{x}}\frac{1}{|S_{n}|}\int_{S_{n}}%
\operatorname{div}(\operatorname{curl}\operatorname{curl}L_{\lambda}[\vec
{j}])\,d\vec{y}=0.
\end{align*}
Hence by the regularity of $\vec{\varphi}$ and $\vec{j}$ we obtain that
$\vec{E}^{\ast}|_{\partial\Omega}\times\overrightarrow{n}\in\mathcal{F}%
_{\text{Div}}(\partial\Omega)$. See \cite[Th. 4.16]{Colton1983} for the
uniqueness of the solution.
\end{proof}

\section{Maxwell's equations in chiral media}

\label{sec:Maxwell-chiral} In this section we apply the general solution of
the system \eqref{system1} provided by Theorem \ref{th:solution-div-curl} to
Maxwell's equations in chiral media. The concept of chirality has played an
important role in chemistry, optics, among others fields (see, e.g.,
\cite{Jaggar1979,Lakhta1989}). Consider the corresponding Maxwell equations
\begin{align}
\operatorname{curl}\overline{E}  &  =i\omega\mu(\overline{H}+\beta
\operatorname{curl}\overline{H}),\label{eq:Maxwell-chiral}\\
\operatorname{curl}\overline{H}  &  =-i\omega\epsilon(\overline{E}%
+\beta\operatorname{curl}\overline{E})+\overline{j},\nonumber\\
\operatorname{div}\overline{E}  &  =\frac{\rho}{\epsilon},\quad
\operatorname{div}\overline{H}=0,\nonumber
\end{align}
where $\beta$ is the chirality measure of the medium. Following the notation
and results of \cite{Kira2001,Krav2003}, we denote
\begin{equation}
\overline{E}=-\sqrt{\mu}\vec{E},\quad\overline{H}=\sqrt{\epsilon}\vec{H}%
\quad\text{and\quad}\overline{j}=\sqrt{\epsilon}\vec{j}
\label{Notation vectors}%
\end{equation}
and consider the following purely vectorial biquaternion valued functions
\[
\vec{\varphi}=\vec{E}+i\vec{H},\quad\vec{\psi}=\vec{E}-i\vec{H}.
\]
Then the system \eqref{eq:Maxwell-chiral} can be written in a biquaternionic
form as follows
\begin{align}
\left(  D+\frac{\lambda}{1+\lambda\beta}\right)  \vec{\varphi}  &  =\frac
{\rho}{\epsilon\sqrt{\mu}}+\frac{i\vec{j}}{1+\lambda\beta}%
,\label{eq:chiral-biquaternionic-1}\\
\left(  D-\frac{\lambda}{1-\lambda\beta}\right)  \vec{\psi}  &  =\frac{\rho
}{\epsilon\sqrt{\mu}}-\frac{i\vec{j}}{1-\lambda\beta},
\end{align}
where $\lambda=\omega\sqrt{\epsilon\mu}$. Or equivalently,
\begin{equation}
\left(  D+\alpha_{1}\right)  \vec{\varphi}=\frac{i}{\lambda}%
(-\operatorname{div}\vec{j}+\alpha_{1}\vec{j}),\quad\left(  D-\alpha
_{2}\right)  \vec{\psi}=-\frac{i}{\lambda}(\operatorname{div}\vec{j}%
+\alpha_{2}\vec{j}), \label{D+alpha1}%
\end{equation}
where the new wave numbers $\alpha_{1}=\lambda/(1+\lambda\beta)$ and
$\alpha_{2}=\lambda/(1-\lambda\beta)$ physically correspond to the propagation
of waves of opposing circular polarizations. This reduction to of
(\ref{eq:Maxwell-chiral}) to a couple of equations for sometimes called
Beltrami fields is often used in relation with electromagnetics in chiral
media (see, e.g., \cite{Athanasiadis}).

\begin{theorem}
Let $\Omega$ be a bounded domain with a Liapunov boundary and $\vec{j}\in
W^{p,\operatorname{div}}(\Omega)$, $1<p<\infty$. Then a general weak solution
of the Maxwell system \eqref{eq:Maxwell-chiral} is given by
\begin{align}
\overline{E}  &  =-\frac{i}{\omega\epsilon}\left(  2\overline{j}-\alpha
_{1}\operatorname{curl}L_{\alpha_{1}}[\overline{j}]+\alpha_{2}%
\operatorname{curl}L_{\alpha_{2}}[\overline{j}]+\operatorname{curl}%
\operatorname{curl}(L_{\alpha_{1}}+L_{\alpha_{2}})[\overline{j}]\right)
\nonumber\label{eq:solution-Maxwell-chiral}\\
&  -\sqrt{\mu}(\vec{u}_{\alpha_{1}}+\vec{u}_{-\alpha_{2}}),\nonumber\\
\overline{H}  &  =-\frac{1}{\lambda}\left(  \alpha_{1}\operatorname{curl}%
L_{\alpha_{1}}[\overline{j}]+\alpha_{2}\operatorname{curl}L_{\alpha_{2}%
}[\overline{j}]-\operatorname{curl}\operatorname{curl}(L_{\alpha_{1}%
}-L_{\alpha_{2}})[\overline{j}]\right) \nonumber\\
&  +\frac{\sqrt{\epsilon}}{i}(\vec{u}_{\alpha_{1}}-\vec{u}_{-\alpha_{2}}),
\end{align}
where $\vec{u}_{\alpha_{1}}$ and $\vec{u}_{-\alpha_{2}}$ are arbitrary
force-free fields associated to the wave numbers $\alpha_{1}$ and $-\alpha
_{2}$, respectively.
\end{theorem}

\begin{proof}
Notice that if $\vec{j}\in W^{p,\operatorname{div}}(\Omega)$, then the $L^{p}$
functions on the right hand side of (\ref{D+alpha1}) belong to
$\operatorname{Sol}_{\alpha_{1}}^{p}(\Omega)$ and $\operatorname{Sol}%
_{-\alpha_{2}}^{p}(\Omega)$, respectively. Thus, by Theorem
\ref{th:solution-div-curl} we have that
\[
\vec{\varphi}=\frac{i}{\lambda}(\vec{j}-\operatorname{curl}\overrightarrow
{T}_{\!\!2,\alpha_{1}}[\vec{j}])+\vec{u}_{\alpha_{1}}\quad\vec{\psi}=\frac
{i}{\lambda}(\vec{j}-\operatorname{curl}\overrightarrow{T}_{\!\!2,-\alpha_{2}%
}[\vec{j}])+\vec{u}_{-\alpha_{2}},
\]
are weak solutions of (\ref{D+alpha1}), where $\vec{u}_{\alpha_{1}}$ and
$\vec{u}_{-\alpha_{2}}$ are arbitrary force-free fields associated to the wave
numbers $\alpha_{1}$ and $-\alpha_{2}$, respectively.

Since $\vec{E}=\frac{1}{2}\left(  \vec{\varphi}+\vec{\psi}\right)  $ and
$\vec{H}=\frac{1}{2i}(\vec{\varphi}-\vec{\psi})$ we have
\begin{align*}
\vec{E}  &  =\frac{i}{\lambda}\left(  2\vec{j}-\operatorname{curl}%
(\overrightarrow{T}_{\!\!2,\alpha_{1}}+\overrightarrow{T}_{\!\!2,-\alpha_{2}%
})[\vec{j}]\right)  +\vec{u}_{\alpha_{1}}+\vec{u}_{-\alpha_{2}},\\
\vec{H}  &  =-\frac{1}{\lambda}\operatorname{curl}(\overrightarrow
{T}_{\!\!2,\alpha_{1}}-\overrightarrow{T}_{\!\!2,-\alpha_{2}})[\vec{j}%
]+\frac{1}{i}(\vec{u}_{\alpha_{1}}-\vec{u}_{-\alpha_{2}}).
\end{align*}
Due to Proposition \ref{prop:identities-Teodorescu} and
(\ref{Notation vectors}) we obtain \eqref{eq:solution-Maxwell-chiral}.
\end{proof}

\section{Conclusions}

A right inverse operator for the operator $\operatorname{curl}+\lambda$ is
constructed for any bounded domain with a Liapunov boundary and as a corollary
a convenient representation for a general weak solution of equation
(\ref{main eq Intro}) is presented. Several applications of this result are
developed which include the nonhomogeneous time-harmonic Maxwell system for
achiral and chiral media. No doubt that the restriction on the smoothness of
the boundary can be weakened. Not less interesting would be obtaining an
analogous result for unbounded domains. The main result of the present work
admits a natural extension onto the $n$-dimensional situation with the aid of
analogous Clifford analysis' tools.

\end{document}